\documentclass[11pt,bezier]{article}
\usepackage{amsmath,amssymb,amsfonts,euscript}
\usepackage{algorithmic, algorithm}
\usepackage{graphicx}
\usepackage{xcolor}

\usepackage[ansinew]{inputenc}
\usepackage{graphicx}
\usepackage{color}
\usepackage{mathrsfs}

\usepackage[colorlinks]{hyperref}
\usepackage[active,new,noold,marker]{xrcs}
\catcode`\=13
\def{$\bowtie$}

\usepackage{eurosym}
\usepackage{lscape} 

\newtheorem{prethm}{{\bf Theorem}}

\newenvironment{thm}{\begin{prethm}{\hspace{-0.5
               em}{\bf}}}{\end{prethm}}

\newtheorem{prepro}{{\bf Theorem}}

\newtheorem{preprop}{{\bf Proposition}}

\newtheorem{precor}{{\bf Corollary}}

\newenvironment{cor}{\begin{precor}{\hspace{-0.5
               em}{\bf}}}{\end{precor}}

\newtheorem{preconj}{{\bf Conjecture}}

\newenvironment{conj}{\begin{preconj}{\hspace{-0.5
               em}{\bf}}}{\end{preconj}}

\newtheorem{predefi}{{\bf Definition}}

\newtheorem{preremark}{{\bf Remark}}

\newenvironment{remark}{\begin{preremark}\rm{\hspace{-0.5
               em}{\bf}}}{\end{preremark}}

\newtheorem{preexample}{{\bf Example}}

\newtheorem{prelem}{{\bf Lemma}}

\newenvironment{lem}{\begin{prelem}{\hspace{-0.5
               em}{\bf}}}{\end{prelem}}

\newtheorem{prelam}{{\bf Lemma}}

\newtheorem{preprob}{{\bf Problem}}

\newenvironment{prob}{\begin{preprob}{\hspace{-0.5
               em}{\bf.}}}{\end{preprob}}

\newtheorem{preproof}{{\bf Proof}}

\newenvironment{proof}[1]{\begin{preproof}{\rm
               #1}\hfill{$\Box$}}{\end{preproof}}

\newtheorem{preali}{{\bf Proof of Theorem 1.}}

\newenvironment{ali}[1]{\begin{preali}{\rm
               #1}\hfill{$\Box$}}{\end{preali}}

\newtheorem{prealii}{{\bf Proof of Theorem 2.}}

\newenvironment{alii}[1]{\begin{prealii}{\rm
               #1}\hfill{$\Box$}}{\end{prealii}}

\newtheorem{prealiii}{{\bf Proof of Theorem 3.}}

\newenvironment{aliii}[1]{\begin{prealiii}{\rm
               #1}\hfill{$\Box$}}{\end{prealiii}}

\newtheorem{prealiiii}{{\bf Proof of Theorem 4.}}

\newenvironment{aliiii}[1]{\begin{prealiiii}{\rm
               #1}\hfill{$\Box$}}{\end{prealiiii}}

\newtheorem{prealij}{{\bf Proof of Theorem 5.}}

\newenvironment{alij}[1]{\begin{prealij}{\rm
               #1}\hfill{$\Box$}}{\end{prealij}}

\newtheorem{prealijj}{{\bf Proof of Theorem 6.}}

\newenvironment{alijj}[1]{\begin{prealijj}{\rm
               #1}\hfill{$\Box$}}{\end{prealijj}}

\newtheorem{prealijjj}{{\bf Proof of Theorem 7.}}

\newenvironment{alijjj}[1]{\begin{prealijjj}{\rm
               #1}\hfill{$\Box$}}{\end{prealijjj}}

\newtheorem{prealijjjk}{{\bf Proof of Theorem 8.}}

\newenvironment{alijjjk}[1]{\begin{prealijjjk}{\rm
               #1}\hfill{$\Box$}}{\end{prealijjjk}}

\title{On the algorithmic complexity of decomposing graphs into regular/irregular
structures}

\author{{\normalsize
{Arash Ahadi${}^{\mathsf{a}}$},\,
{Ali Dehghan${}^{\mathsf{b}}$},\,
{ Mohammad-Reza Sadeghi${}^{\mathsf{c}}$},\,
{ Brett Stevens${}^{\mathsf{d}}$}\,
}\vspace{3mm}
\\
{\footnotesize{${}^{\mathsf{a}}$\it Department of
Mathematical Sciences, Sharif University of Technology, Tehran,
Iran}}
{\footnotesize{}}\\
{\footnotesize{${}^{\mathsf{b}}$\it
Systems and Computer Engineering Department, Carleton University, Ottawa,   Canada}}
{\footnotesize{}}\\
{\footnotesize{${}^{\mathsf{c}}$\it
Department of Mathematics and Computer Science,
Amirkabir University of Technology, Tehran, Iran}}
{\footnotesize{}}\\
{\footnotesize{${}^{\mathsf{d}}$\it
School of Mathematics and Statistics, Carleton University, Ottawa, Canada }}
\thanks{{\it E-mail addresses}:  $\mathsf{arash\_ahadi@mehr.sharif.edu}$ (Arash Ahadi), $\mathsf{alidehghan@sce.carleton.ca}$ (Ali Dehghan), $\mathsf{msadeghi@aut.ac.ir}$ (Mohammad-Reza Sadeghi), $\mathsf{brett@math.carleton.ca}$ (Brett Stevens). } }

\date{}
\begin{document}
\maketitle

\begin{abstract}
{\small \noindent
A {\em locally irregular} graph is a graph whose adjacent vertices have
distinct degrees, a {\em regular} graph is a graph where each vertex has the same degree and
a {\em locally regular} graph is a graph where for
every two adjacent vertices $u,v$, their degrees are equal.
In this paper, we investigate the set of all problems which are related to decomposition of graphs into regular, locally regular
and/or locally irregular
subgraphs and we present some
polynomial time algorithms, {\bf NP}-completeness results, lower bounds and upper bounds for them.
Among our results,  one of our lower bounds makes  use of  mutually orthogonal Latin squares which is  relatively novel.
}

\begin{flushleft}

\end{flushleft}

\end{abstract}

\section{Introduction}
\label{}

For a  family $\mathcal{F}$ of graphs, an {\it $\mathcal{F}$-decomposition}
of a graph $G$ is a decomposition of the edge set of the graph $G$ into subgraphs isomorphic
to members of the family $\mathcal{F}$. Note that the family $\mathcal{F}$ of graphs can be anything, for instance, all regular graphs or all complete graphs.  During the last decade, the computational complexity of this problem has received a considerable
attention. For example, Holyer proved that it is {\bf NP}-hard to decompose the edges of a graph
into the minimum number of complete subgraphs \cite{MR635429}. For more examples see \cite{ahadi2017algorithmic,  MR2108396, MR2556522, MR3017971} and the references therein.

A {\it locally irregular} graph is a graph whose adjacent vertices have
distinct degrees. Also, a {\it regular} graph is a graph where each vertex has the same degree
and a {\it locally regular} graph is a graph where for
every two adjacent vertices $u$ and $v$, their degrees are equal.
In this work, we consider the set of all problems which are related to decomposition of graphs into regular and/or locally irregular
subgraphs and
we present some
polynomial time algorithms, {\bf NP}-completeness results, upper bounds and lower bounds for them.
A summary of our results and open problems are shown in Table 1 and Table 2. These graph families have received attention recently because of their relationship to the 1-2-3 Conjecture \cite{MR2047539}. (For more information about the 1-2-3 Conjecture and it variations see \cite{MR3478612, MR3624793,   MR3315373, MR3589724, MR2881497}  and the references therein.)

Before we start we would like to draw the readers attention to the different
kinds of subgraphs that are suitable parts of a decomposition in this paper.  When a
decomposition has a {\em locally irregular} component, this subgraph, $G_i$ has the
property that if $u$ and $v$ are adjacent in $G_i$ then their degrees in $G_i$ must
be different, $d_i(u) \neq d_i(v)$.  A component, $G_i$ is a {\em regular} subgraph if
$d_i(u)$ is constant for all vertices in $V(G_i)$.  These two are the quite standard and
we have restated them for completeness and clarity.  Two other kinds of component subgraphs
are investigated at various points in this paper which are less traditional and we define
them here so that when the reader encounters them herein, she will have already seen the
distinction.  A component, $G_i$ of an edge decomposition is {\em locally regular} if for
all adjacent $u, v \in V(G_i)$, their degrees in $G_i$ must be equal, $d_i(u) = d_i(v)$.  Clearly
all regular graphs are locally regular but the disjoint union of a cycle and an edge is locally
regular without being regular.  The final type of component allowed is one in which each component
is permitted to be either regular or locally irregular.  Once again a regular, locally regular or
locally irregular  graph fits this criterion but the disjoint union of a cycle, an edge and a non-trivial
star satisfies the criterion and is neither regular, locally regular nor locally irregular.

\section{Our results and motivations}

\subsection{Locally irregular graphs}

Motivated by  the fact that every connected graph with at least two vertices contains a pair of vertices of the same degree and the 1-2-3 Conjecture, we consider
the locally irregular graphs.
In 2015, Baudon {\it et al.} introduced the notion
of decomposition into locally irregular subgraphs, where by a decomposition
they mean a partitioning of the edges \cite{baudon2013decomposing}.
In such a case, we want to decompose the graph $G$ into
locally irregular subgraphs, where by a decomposition of the graph $G$ into $k$ locally
irregular subgraphs
 we refer to a decomposition $E_1, \ldots, E_k$ of $E(G)$ such
that the subgraph $G[E_i]$ is locally irregular for every $i =1,\ldots,k$.
The {\it irregular chromatic index} of the graph $G$, denoted by $ \chi'_{irr}$,
is the minimum number $k$ such that the graph $G$ can be decomposed into $k$ locally
irregular subgraphs.

Baudon {\it et al.}
identified all connected graphs which cannot be decomposed into locally
irregular subgraphs and  call them {\it exceptions} \cite{baudon2013decomposing}.
They conjectured that apart from these exceptions all other connected
graphs can be decomposed into three locally irregular subgraphs \cite{baudon2013decomposing}.

\begin{conj} \label{C1} {\it \cite{baudon2013decomposing}}
For each non-exception graph $G$, we have $\chi'_{irr}(G)\leq 3$.
\end{conj}

Afterwards, Bensmail {\it et al.} proved that  every bipartite graph $G$ which is not an odd length path satisfies $\chi'_{irr}(G)\leq 10$ \cite{bensmail2016decomposing}. Recently, Lu{\v{z}}ar {\it et al.} improved   the upper bound for bipartite graphs and general graphs, into 7 and 220, respectively \cite{luvzar2016new}.

From another point of view, Bensmail and Stevens considered the problem of decomposing the edges of graph into
some subgraphs, such that {\bf in each subgraph every component} is either
regular
or locally irregular \cite{bensmail2014edge}. The {\it regular-irregular chromatic index} of graph $G$,
 denoted by $ \chi'_{reg-irr}$,
is the minimum number $k$ such that $G$ can be decomposed into $k$ subgraphs, such that each component  of every
subgraph is locally
irregular or regular \cite{bensmail2014edge}.
They conjectured that the edges of every graph can be
decomposed into at most two subgraphs, such that each component of every subgraph
 is regular or locally irregular \cite{bensmail2014edge}.

\begin{conj} \label{C2}  {\it \cite{bensmail2014edge}}
For each graph $G$, we have $\chi'_{reg-irr}(G)\leq 2$.
\end{conj}

How much easier is Conjecture \ref{C1} if we relax the problem
and only require that each subgraph (instead of each component)
should be  locally irregular or regular?
With this motivation in mind,
we consider  the problem of partitioning the edges of graph into subgraphs, such that each subgraph is
regular or locally irregular.
The {\it regular-irregular number} of graph $G$, denoted by $reg-irr(G)$,
is the minimum number $k$ such that the graph $G$ can be decomposed into $k$ subgraphs, such that each
subgraph is locally irregular or regular.

\begin{equation}
\chi'_{reg-irr}(G) \leq reg-irr(G) \leq \chi'_{irr}(G).
\end{equation}

Motivated by Conjecture \ref{C1} and Conjecture \ref{C2}, we present the following conjecture.
With Conjecture \ref{C3} we weaken Conjecture \ref{C1} and strengthen Conjecture \ref{C2}.

\begin{conj} \label{C3}
Each  graph can be decomposed into $3$ subgraphs, such that each
subgraph is locally irregular or regular.
\end{conj}

There are infinitely many graphs such that their regular-irregular numbers are three.
For example, consider the following tree. First, join two
vertices $v$ and $u$ by an edge. Then consider four paths of lengths 6,6,2,2  called  $P_1, P_2, P_3, P_4$,
respectively.
Identify one of the ends for each of $P_1$, $P_2$ with $v$, and identify one of the ends for each of $P_3$, $P_4$
with $u$. Call the resultant tree $\mathcal{T}$. It is easy to check that the tree $\mathcal{T}$ cannot be decomposed
into two subgraphs, such that each
subgraph is locally irregular or regular.
We show that
deciding whether a given planar bipartite graph $G$ with maximum degree three can be decomposed into at most
two subgraphs, such that each subgraph is regular or locally irregular is {\bf NP}-complete.

\begin{thm}\label{T1}
Determining whether the regular-irregular number of a
given planar bipartite graph $G$ with maximum degree three is at most two, is {\bf NP}-complete.
\end{thm}

From the proof of Theorem \ref{T1}, one can obtain the following corollary.

\begin{cor} \label{cor1}
For a given planar bipartite graph $G$ with maximum degree three, deciding whether the edge set of $G$ can
be decomposed into two subgraphs $\mathcal{R}$ and $\mathcal{I}$ such that $\mathcal{R}$ is
 regular and $\mathcal{I}$ is locally irregular, is {\bf NP}-complete.
\end{cor}

If $T$ is a tree which is not an odd length path, then its irregular chromatic index is at most three
 and there exist infinitely many trees with irregular chromatic
index three \cite{baudon2013decomposing}. Baudon {\it et al.} proved that the problem of determining
the irregular chromatic index of a graph can be handled in linear
time when restricted to trees  and if $T$ is a tree with $\Delta(T)> 4$, then its
irregular chromatic index is at most two \cite{bensmail2013complexity}.
Afterwards, Bensmail and Stevens proved that if $T$ is a tree, then its regular-irregular chromatic index is at most two
\cite{bensmail2014edge}.

Here, for every $k>2$, we construct a tree $T$ with $\Delta(T)= k$ such that $T$ cannot be decomposed into a matching and
a locally irregular subgraph and also, we show that every tree
can be decomposed into two matchings and a locally irregular subgraph.

\begin{thm}\label{newt1}\\
(i) For every $k>2$, there is a tree with $\Delta(T)= k$ such that $T$ cannot be decomposed into a matching and
a locally irregular subgraph.\\
(ii) Every tree can be decomposed into two subgraphs
$\mathcal{P}$ and $\mathcal{R}$ such that $\mathcal{R}$ is
a matching and each component of $\mathcal{P}$ is an edge or a locally irregular component.\\
(iii) Every tree can be decomposed into two matchings and a locally irregular subgraph.
\end{thm}

In \cite{bensmail2013complexity}, Baudon {\it et al.} proved that determining whether a given planar graph $G$,
can be decomposed into two
locally irregular subgraphs is {\bf NP}-complete.
But their reduction does not preserve the planarity. In this paper, by another reduction, we show
 that determining whether a given planar graph $G$,
can be decomposed into two locally irregular subgraphs is {\bf NP}-complete.

\begin{thm}\label{T2} Determining whether the irregular chromatic index of a
given planar  graph $G$ is at most two, is {\bf NP}-complete.
\end{thm}

There is an interesting connection between an edge-labeling which
is an additive vertex-coloring and the irregular chromatic index of regular graphs.

\begin{remark} \label{R1}
Karo\'nski, \L{}uczak and Thomason
initiated the study of edge-labelings which
give additive vertex-colorings. That means for every edge
$uv$, the sum of labels of the edges incident to $u$ is different
from the sum of labels of the edges incident to $v$ \cite{MR2047539}.
Dudek and Wajc
 showed that determining whether a given graph has an edge-labeling which
is an additive vertex-coloring from $\{1,2\}$  is
$ \mathbf{NP} $-complete \cite{David}. Afterwards, Ahadi {\it et al.} proved that determining whether a given 3-regular graph $G$ has
an edge-labeling which
is an additive vertex-coloring from $\{1,2\}$  is
$ \mathbf{NP} $-complete \cite{MR3072733}. For a given 3-regular graph $G$, it is easy to see that the graph $G$  has an edge-labeling which
is an additive vertex-coloring from $\{1,2\}$ if and only if the edge set of the graph $G$ can be decomposed
 into at most two locally irregular subgraphs. Thus, for a given 3-regular graph $G$, deciding whether $\chi'_{irr}(G)=2$
 is {\bf NP}-complete \cite{MR3072733}.
\end{remark}

\subsection{Regular graphs}

The edge set of every graph can be decomposed
such that the subgraph induced by each subset is regular (to obtain a trivial upper bound consider the case that each subgraph is a matching).
In 2001, Kulli {\it et al.} introduced the {\it regular number} of  graphs \cite{kulli}.
The {\it regular number} of a graph $G$, denoted by $reg(G)$, is the minimum number of subsets into which the
edge set of the graph  $G$ can be decomposed
so that the subgraph induced by each subset is regular.
Nonempty subsets $E_1, \ldots ,E_r$ of $E(G)$ are said to form a regular
decomposition of the graph $G$ if the subgraph induced by each subset is regular.
The  edge chromatic number
of a
graph, denoted by $\chi '(G)$, is the minimum size of a decomposition of the edge set into $1$-regular subgraphs.
By Vizing's theorem  the edge chromatic number of a graph $G$ is equal to either $ \Delta(G) $ or $ \Delta(G) +1 $ (see \cite{MR1367739}, page 197).
Hence, the regular number problem is a generalization of the  edge chromatic number and we have the following:

\begin{equation}
 reg(G)\leq \chi '(G) \leq \Delta(G) +1.
\end{equation}

Determining whether  $reg(G)\leq \Delta(G)$ holds for all
connected graphs was posed an open problem in \cite{kulli2}.
It was shown   that not only  there exists a counterexample for the above bound but also for a given connected graph $G$ deciding whether $reg(G)\leq \Delta(G)$ is
{\bf NP}-complete \cite{reg}.
Designing an algorithm to decompose a given bipartite graph into the minimum number
of regular subgraphs was posed as another problem in \cite{kulli2}.
But, it was proved that computation of the regular number is {\bf NP}-hard for connected bipartite graphs. Also,
it was proved that deciding whether $reg(G) = 2$ for a given connected 3-colorable
graph $G$ is {\bf NP}-complete \cite{reg}.
Here, we improve the previous results and show that for a given bipartite graph $G$ with maximum degree six,
determining whether $ reg(G) = 2 $  is {\bf NP}-complete. Furthermore, we show that
there is polynomial time algorithm to decide whether $ reg(G) = 2 $ for a given graph $ G $
with maximum degree five.

\begin{thm}\label{T3}\\
(i) For every number $\alpha\geq 3$,
determining whether $ reg(G) = 2 $ for a given bipartite graph $G$ with degree set $\{\alpha, 2\alpha\}$,
is {\bf NP}-complete.\\
(ii) There is polynomial time algorithm to decide whether $ reg(G) = 2 $ for a given graph $ G $
with maximum degree five.
\end{thm}

Also, we consider the problem of determining the regular number for planar graphs. Note that
every planar graph $G$ with degree set $\{2, 4\}$ can be decomposed into two regular subgraphs (see the proof of part (ii)
of Theorem \ref{T3}).

\begin{thm}\label{T4}
 Determining whether $ reg(G) = 2 $ for a given planar graph $G$ with degree set $\{3, 6\}$,
is {\bf NP}-complete.
\end{thm}

\subsection{Locally regular graphs}

We say that a graph $G$ is {\it locally regular} if each component of $G$ is regular (Note that a regular graph is locally regular
but the converse does not hold).
The {\it  regular chromatic index} of a graph $G$ denoted by $\chi'_{reg}$ is the minimum number of subsets into which the
edge set of $G$ can be decomposed so that the subgraph induced by each subset is locally regular.
From the definitions of locally regular and regular graphs we have the following bound.

\begin{equation}
\chi'_{reg}(G) \leq reg(G)\leq  \Delta(G)+1.
\end{equation}

It was shown that determining whether $reg(G)\leq \Delta(G)$ for a given
connected graph $ G $  is {\bf NP}-complete \cite{reg}. Here, we show that every graph $G$ can be decomposed into $\Delta(G)$
subgraphs such that each subgraph is locally regular.
 We use the concept of semi-coloring to prove that fact. Daniely and Linial defined a semi-coloring of graphs for the investigation of the
tight product of graphs \cite{daniely2012tight}. Afterwards, Furuya {\it et al.} proved that every graph has
a semi-coloring \cite{furuya2014existence}.

\begin{thm}\label{newt2}
For every graph $G$, $\chi'_{reg}(G)\leq \Delta(G)$ and this bound is sharp for trees.
\end{thm}

\begin{remark}
The difference between the regular number and the regular chromatic index of a graph can be arbitrary large. For
a fixed $t$, consider a copy of the complete graphs $K_1, \ldots, K_t$. For each $i$, $i=1,\ldots, t-1$, join one of the vertices
of the complete graph $K_i$ to one of the vertices of the complete graph $K_{i +1}$. Call the resulting graph $G$. If we put the set of
edges of complete graphs $K_1, \ldots, K_t$ in one subgraph and the other edges of $G$ in an another subgraph, we obtain an
edge decomposition of $G$ into two
locally regular subgraphs, so $\chi'_{reg}(G)\leq 2$. On the other hand, since $G$ has $t + \mathcal{O}(1) $
different numbers in its degree
set, $reg(G)\geq \lg t$.
\end{remark}

Suppose that $G$ is a connected graph with degree set $\{\alpha, 2\alpha\}$ and the induced graph on the
set of vertices of degree $2\alpha$ forms an independent set. It is easy to check that $\chi'_{reg}(G)=2$
if and only if $reg(G) = 2$ (since the induced graph on the
set of vertices of degree $2\alpha$ forms an independent set, so in each decomposition every component is $\alpha$-regular).
Thus, by the proof of Theorem \ref{T3} and Theorem \ref{T4}, we have the following
corollary.

\begin{cor}\label{cor2}
{\\
(i) Determining whether $\chi'_{reg}(G)=2$ for a given bipartite graph $G$ with maximum degree
six, is {\bf NP}-complete.\\
(ii) Determining whether $\chi'_{reg}(G)=2$ for a given planar graph $G$ with degree set
$\{3, 6\}$, is {\bf NP}-complete.
}\end{cor}

In Theorem \ref{T3}, we prove that there is polynomial time algorithm to decide whether $ reg(G) = 2 $ for a given graph $ G $
with maximum degree five. Here, we show that deciding whether a given  subcubic graph can be decomposed into two
subgraphs such that each subgraph is locally regular, is {\bf NP}-complete.

\begin{thm} \label{newt3}
For a given subcubic graph $G$, determining whether it can be decomposed into two
subgraphs such that each subgraph is locally regular, is {\bf NP}-complete.
\end{thm}

\subsection{Locally $k$-irregular graphs}

We say that a graph is {\it locally $k$-irregular} if and only if for every two adjacent vertices $v$ and $u$, $|d(v)-d(u)|\geq k$.
In the following, we would like to decompose $G$ into
locally $k$-irregular subgraphs, where by a decomposition of $G$ into $t$ locally
$k$-irregular subgraphs
 we refer to a partition $E_1, \ldots, E_t$ of $E(G)$ such
that $G[E_i]$ is locally $k$-irregular for every $i =1,\cdots, t$.
The {\it $k$-irregular chromatic index} of $G$, denoted by $ \chi'_{k-irr}$,
is the minimum number $t$ such that $G$ can be decomposed into $t$ locally
$k$-irregular subgraphs. For every $k$, let $\mathcal{G}_k$ be the set of graphs which can be decomposed into locally
$k$-irregular subgraphs, define:

\begin{center}
$h(k)=\max_{ G\in \mathcal{G}_k}  \chi'_{k-irr}(G) $.
\end{center}

Baudon {\it et al.}
characterized all connected graphs which cannot be decomposed into locally
$1$-irregular subgraphs and  called them {\em exceptions} \cite{baudon2013decomposing}.
They conjectured $h(1)\leq 3$ \cite{baudon2013decomposing}.
Also, Baudon {\it et al.} asked the computational complexity of
determining $\chi'_{1-irr}=2$ for bipartite graphs \cite{bensmail2013complexity}.
We show that for
each $k>1 $, deciding whether $\chi'_{k-irr}(G)=2$ for a given planar bipartite graph $G$  is ${\bf NP}$-complete.
For all $k$ we prove the lower bound $h(k) \geq 2k+1$ and we will use
 mutually orthogonal Latin squares to prove that  $h(k)=\Omega(k^2)$. Finding a better lower bound can
be interesting for future work.

\begin{thm}\label{T5}\\
(i) For every $k>1 $, determining whether $\chi'_{k-irr}(G)=2$ for a given planar bipartite graph $G$  is ${\bf NP}$-complete.\\
(ii) For each $k$, $h(k) \geq 2k+1$ and if $k>3$, $h(k) \geq 4k$.\\
(iii) $h(k)=\Omega(k^2)$.\\
(iv) For each fixed $k$, there is a polynomial time algorithm to decide whether a given graph $G$
 with maximum degree $k+1$ can be decomposed into two locally $k$-irregular subgraphs.
\end{thm}

\subsection{Summary of results}

A summary of results and open problems are shown in Table 1 and Table 2. In the first table
we summarize the recent results on the computational complexity of deciding
whether a family of graphs can be decomposed into two subgraphs with some conditions and
in the second table
we summarize the recent upper bounds and conjectures on the different types of partitioning.
For more information about the decomposing the graphs into regular/irregular subgraphs see \cite{MR3581338, MR3512668, MR3543192, luvzar2016new, merker2016graph}.

\begin{table}[ht]
\small
\caption{Recent results on edge  decomposing of graphs into two subgraphs. Blue text shows our results.} 
\centering 
\begin{tabular}{l l  l l l} 
\hline\hline 
                 & Tree  &  Bipartite  & Planar  & Subcubic\\ [0.5ex] 
\hline 
Irregular chromatic index &  {\bf P} \cite{bensmail2013complexity} &  Open \cite{bensmail2013complexity} &
\color{blue} {\bf NP}-c (Th. \ref{T2}) \color{black}
& \color{blue} {\bf NP}-c (Remark \ref{R1}) \color{black}\\
Regular-irregular number &  Open  & \color{blue} {\bf NP}-c (Th.  \ref{T1}) \color{black}
& \color{blue} {\bf NP}-c (Th.  \ref{T1}) \color{black} & \color{blue} {\bf NP}-c (Th.  \ref{T1}) \color{black}\\
1 regular plus 1 irregular &    Open  & \color{blue} {\bf NP}-c (Cor. \ref{C1}) \color{black}
 & \color{blue} {\bf NP}-c (Cor. \ref{C1}) \color{black} & \color{blue} {\bf NP}-c (Cor. \ref{C1}) \color{black}\\
Regular number&  {\bf P} \cite{kulli} &  \color{blue} {\bf NP}-c (Th.  \ref{T3}) \color{black} & \color{blue} {\bf NP}-c (Th.  \ref{T4}) \color{black} &
\color{blue} {\bf P} (Th. \ref{T3}) \color{black} \\
Regular chromatic index & \color{blue}  {\bf P} (Th. \ref{newt2}) \color{black} & \color{blue} {\bf NP}-c (Cor. \ref{cor2}) \color{black} & \color{blue} {\bf NP}-c (Cor. \ref{cor2}) \color{black}&
\color{blue} {\bf NP}-c (Th.  \ref{newt3}) \color{black} \\
$k$-irregular chromatic index $(k>1)$&  Open &  \color{blue}  {\bf NP}-c (Th.  \ref{T5}) \color{black} & \color{blue} {\bf NP}-c (Th.  \ref{T5}) \color{black} & \color{blue} {\bf P}  (Th.  \ref{T5}) \color{black}\\
Regular-irregular chromatic index & {\bf P}  \cite{bensmail2014edge} &
 {\bf P}  (Conj. \ref{C2} \cite{bensmail2014edge})   &
  {\bf P} (Conj. \ref{C2} \cite{bensmail2014edge})    &
  {\bf P}  (Conj. \ref{C2} \cite{bensmail2014edge})   \\
 [1ex] 
\hline 
\end{tabular}
\label{table:nonlin} 
\end{table}

\begin{table}[ht]
\small
\caption{Recent upper bounds and conjectures} 
\centering 
\begin{tabular}{l c  c c } 
\hline\hline 
 & Tree &  Bipartite graphs &   General graphs\\ [0.5ex] 
\hline
Irregular chromatic index &    3 \cite{baudon2013decomposing} &
   3 (Conj. \ref{C1} \cite{baudon2013decomposing})
&   3 (Conj. \ref{C1} \cite{baudon2013decomposing})   \\
Regular-irregular chromatic index &   2  \cite{bensmail2014edge}
& 6 \cite{bensmail2014edge} &
   2 (Conj. \ref{C2} \cite{bensmail2014edge})  \\
Regular-irregular number &   3  \cite{baudon2013decomposing}   &
   3 (Conj. \ref{C3})
 &  3 (Conj. \ref{C3})  \\
Regular number&   $\Delta(G)$ \cite{kulli}   & $\Delta(G)$ \cite{kulli} &  $\Delta(G)+1$ \cite{kulli}  \\
Regular chromatic index & \color{blue} $\Delta(G)$ (Th.  \ref{newt2}) \color{black}   & \color{blue}$\Delta(G)$ (Th.  \ref{newt2}) &  \color{blue}$\Delta(G)$
(Th.  \ref{newt2}) \color{black} \\
$k$-irregular chromatic index &   $f(k)$ (Prob. \ref{problem1})
&    $f(k)$ (Prob. \ref{problem1})
&   $f(k)$ (Prob. \ref{problem1})  \\
 [1ex] 
\hline 
\end{tabular}
\label{table:nonlin} 
\end{table}

\section{Proofs}

Here, we show that
determining whether a given planar bipartite graph $G$ with maximum degree three can be decomposed into at most
two subgraphs, such that each subgraph is regular or locally irregular is {\bf NP}-complete.

\begin{ali}{
Clearly, the problem is in $ \mathbf{NP} $. We reduce {\em  Cubic Planar 1-In-3 3-Sat} to our problem.
 Moore and Robson \cite{MR1863810} proved
that the following problem is $ \mathbf{NP}$-complete.

 {\em  Cubic Planar 1-In-3 3-Sat.}\\
\textsc{Instance}: A 3-Sat formula $\Phi=(X,C)$
 such that every variable
 appears in exactly three clauses, there
 is no negation in the formula, and the
bipartite graph obtained by linking a variable and a clause if and only
 if the
 variable appears in the clause, is planar.\\
\textsc{Question}: Is there a truth assignment for $X$ such that
 each clause in $C$ has exactly
one true literal?

Let $\Phi=(X,C)$ be an instance of {\em  Cubic Planar 1-In-3 3-Sat}.
 Without loss of generality suppose that the number of clauses in  $\Phi$
 is even and $ C= \{c_i : i \in \mathbb{Z}_m \}$. Note that we have $|X|=|C|=m$. We convert the formula $\Phi$ into a
 graph $G_\Phi$ such that the formula $\Phi$ has a 1-in-3 satisfying assignment if and
 only if  the edge set of the graph $G_\Phi$ can be decomposed into two subgraphs
such that each subgraph is regular or locally irregular.
For every pair $(x,i)$, where $x \in X$ and $i \in \mathbb{Z}_m$, consider two cycles
$x_i x'_i x''_i x'''_i$ and $z_i z'_i z''_i z'''_i$, also, put a vertex $ s_i$ and join the vertex $ s_i$ to the vertices
$x''_i, z''_i$.
 Next
for every number $i$, $i \in \mathbb{Z}_m$ put a vertex $r_i $ and join the vertex $r_i$ to the vertices  $z'''_i, x'''_{(i+1 \mod m)}$. Also, for each clause $c_i$, $c_i\in C$ put a clause vertex $c_i$ and
 for every variable $x\in X$ if $x$ appears in $c_i$, then put the edge $c_i x_i$. Having
  done these for all variables $x$ and all clauses $c_i$, in the resultant graph the
degree of every vertex is  two or three and the graph is planar. Next, for every vertex $v$ of degree two put a new vertex $v_u$
(we will call this new vertex
a dummy vertex in our proof) and join the vertex $v$
to the vertex $v_u$. Call the resultant graph $\mathcal{F}$.  It is easy to check that $\mathcal{F}$ is planar, bipartite and
its degree set is $\{1,3\}$. See Fig. \ref{Fig001}.

\begin{figure}[ht]
\begin{center}
\includegraphics[scale=.35]{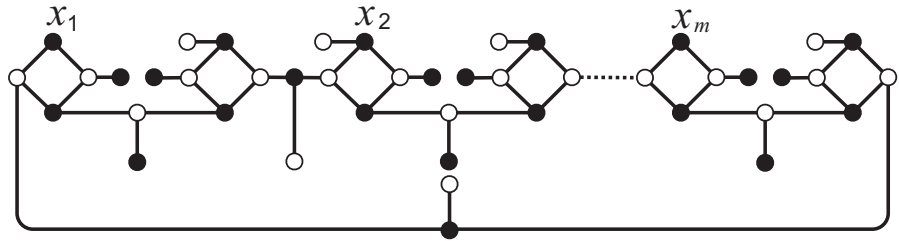}
\caption{A part of the graph $\mathcal{F}$, when  the variable $x$ appears in $c_1$, $c_4$ and $c_m$.
} \label{Fig001}
\end{center}
\end{figure}

Now, consider the following tree. First, join two
vertices $v$ and $u$ by an edge. Then consider four paths of lengths 4,4,2,2  called   $P_1, P_2, P_3, P_4$,
respectively.
Identify one of the ends of each of $P_1$, $P_2$ with $v$, and  identify one of the ends of each of $P_3$, $P_4$
with $u$. Call the resultant tree $\mathcal{T}$.
The graph $\mathcal{T}$ cannot be decomposed into two locally irregular subgraphs. Also,
$\mathcal{T}$ cannot be decomposed into two  regular subgraphs.
Thus, we can only decompose $\mathcal{T}$ into  two subgraphs such that one subgraph is  regular  and the other one
is locally irregular.
Since $\mathcal{T}$ is  a tree and is not locally irregular, the regular subgraph should be 1-regular.

{\bf Construction of $G_\Phi$.}
Consider a copy of the graph $\mathcal{F}$, a copy of the graph $\mathcal{T}$ and call their union $G_\Phi$.
$G_\Phi$  is planar, bipartite and its maximum degree is three. Also, we cannot decompose it into two locally irregular subgraphs
or into two regular subgraphs.

First assume that the graph $G_\Phi$ can be decomposed into a regular subgraph $\mathcal{R}$ and a locally irregular subgraph $\mathcal{I}$.
By the structure of the tree $\mathcal{T}$, the subgraph $\mathcal{R}$ should be 1-regular.
For every pair $(x,i)$, where $x \in X$ and  $i \in \mathbb{Z}_m$, let $y_i$ be the unique neighbor
of the vertex $x_i$ which is not in the cycle $x_i x'_i x''_i x'''_i$ (similarly, for every pair $(x',i)$, where $x \in X$
and  $i \in \mathbb{Z}_m$, let $y'_i$ be the unique neighbor of the vertex $x'_i$ which is not in the
cycle $x_i x'_i x''_i x'''_i$). Note that according to the structure of the graph $G_\phi$, the vertex $y_i$ is a clause vertex
 or a dummy vertex. Also, the vertex $y_i'$ is a dummy vertex.
Since the subgraph $\mathcal{R}$ is 1-regular, the degree
sequence $d_\mathcal{I}(x_i)d_\mathcal{I}(x'_i)d_\mathcal{I}(x''_i)d_\mathcal{I}(x'''_i)$ is 2323 or 3232 (otherwise
 the subgraph  $\mathcal{I}$
is not locally irregular). Similarly, the degree
sequence $d_\mathcal{I}(z_i)d_\mathcal{I}(z'_i)d_\mathcal{I}(z''_i)d_\mathcal{I}(z'''_i)$ is 2323 or 3232.
Consequently, for every $i$, $E(\mathcal{R}) \cap \{x_i y_i \}\neq \emptyset$ or $E(\mathcal{R}) \cap \{  x'_i y'_i \}\neq \emptyset$. On the other hand, for every number $i$, $i \in \mathbb{Z}_m$,  $s_ix''_i,s_iz''_i, z'''_i r_i, r_ix'''_{(i+1 \mod m)} \in E(G_\Phi)$,
therefore $E(\mathcal{R}) \cap \{x_i y_i , x'_i y'_i : i \in \mathbb{Z}_m\}$ is exactly $\{x_i y_i : i \in \mathbb{Z}_m\}$
or $\{x'_i y'_i : i \in \mathbb{Z}_m\}$ ({\bf Property A}).
Also, since $\mathcal{R}$ is 1-regular, for every clause vertex $c_i$, $d_{\mathcal{I}}(c_i)$ is 2 or 3.
But $d_\mathcal{I}(c_i)$ is not 3. To the contrary, assume that $d_\mathcal{I}(c_i)$ is three and
let $u$ be a neighbor of $c_i$ in $\mathcal{I}$. We have $d_{\mathcal{I}}(u)=3$, so $d_\mathcal{I}(c_i)=d_\mathcal{I}(u)$,
but this is a contradiction.
Consequently, $d_{\mathcal{I}}(c_i)=2$ ({\bf Property B}).
Define $\Psi: X \rightarrow \{true, false\}$ such that $\Psi(x)=true$ if and only if $x_i c_i \in E(\mathcal{R})$ for some $i$.
By Property A, the function $\Psi$ is well-defined and by Property B, the function $\Psi$ is a 1-in-3 satisfying assignment for the formula $\Phi$.

On the other hand, assume that the function $\Psi: X \rightarrow \{true, false\}$ is a 1-in-3 satisfying assignment for $\Phi$.
We show that there is a partition for the edge set of $G_\Phi$ into two subgraphs
$\mathcal{R}$ and $\mathcal{I}$ such that $\mathcal{R}$ is regular and
$\mathcal{I}$ is locally irregular.
For every pair $(x,i)$, where $x \in X$ and  $i \in \mathbb{Z}_m$, let $w_i$ ($w_i', w_i'', w_i'''$, respectively) be the unique neighbor
of the vertex $z_i$ ($z_i', z_i'', z_i'''$, respectively) which is not in the cycle $z_i z'_i z''_i  z'''_i$. Similarly, let $y_i''$ ($y_i'''$, respect.) be the unique neighbor
of the vertex $x_i''$ ($x_i'''$, respect.) which is not in the cycle $x_i x'_i x''_i  x'''_i$. Now, for each $x\in X$,
put $\{x_i y_i  , x_i''y_i''  ,z_i'w_i' , z_i'''w_i'''  : i \in \mathbb{Z}_m \}$ in $E(\mathcal{R}) $ if $\Psi(x)=true$ and
put $\{x_i' y_i', x_i'''y_i''',z_i w_i  , z_i'' w_i''  : i \in \mathbb{Z}_m \}$ in $E(\mathcal{R}) $ if $\Psi(x)=false$. One can see that this is a decomposition of the graph $\mathcal{F}$ into  two subgraphs
$\mathcal{R}$ and $\mathcal{I}$ such that $\mathcal{R}$ is 1-regular and
$\mathcal{I}$ is locally irregular. Note that the tree $\mathcal{T}$ can be decomposed into two subgraphs
$\mathcal{R}$ and $\mathcal{I}$ such that $\mathcal{R}$ is 1-regular and
$\mathcal{I}$ is locally irregular. This completes the proof.
}\end{ali}

If $T$ is a tree which is not an odd length path, then its irregular chromatic index is at most three and
also, there exist infinitely many trees with irregular chromatic
index 3 \cite{baudon2013decomposing}.
Bensmail and Stevens proved that if $T$ is a tree, then its regular-irregular chromatic index is at most two.
Here, for every $k>2$, we construct a tree $T$ with $\Delta(T)= k$ such that $T$ cannot be decomposed into a matching and
a locally irregular subgraph and also, we show that every tree
can be decomposed into two matchings and a locally irregular subgraph.

\begin{alii}{
(i)
Let $k>2$ be a fixed number, we construct a tree $\mathcal{T}$ with $\Delta( \mathcal{T} )= k$ such that the tree $\mathcal{T}$
 cannot be decomposed into a matching and
a locally irregular subgraph.
First, consider the following auxiliary tree. Join two
vertices $v$ and $u$ by an edge. Then consider four paths of lengths 3,3,3,1 and call them  $P_1, P_2, P_3, P_4$,
respectively. Identify one of the ends for
 each of $P_1$, $P_2$ with $v$, and finally identify one of the ends for each of $P_3$, $P_4$
with $u$. Call the resultant tree $\mathcal{T'}$. The tree $\mathcal{T'}$ has exactly one vertex of degree one such that its neighbor
has degree three, call this vertex the {\em bad vertex} of the tree $\mathcal{T'}$.
Now, consider $k$ copies of the tree $\mathcal{T'}$ and a new vertex $z$. Join the vertex $z$ to the bad vertex of each copy of
$\mathcal{T'}$ and call the resultant tree $\mathcal{T}$. If the tree $\mathcal{T}$ can be decomposed into a matching and
a locally irregular subgraph, by the structure of $\mathcal{T'}$, in each copy of
$\mathcal{T'}$, the edge between the bad vertex of $\mathcal{T'}$ and $z$ should be in matching, but since the degree
of $z$ is $k> 2$, this is a contradiction.\\ \\
(ii)
The proof is by induction on the number of edges in the tree.  Assume that, for some integer $m\geq 2$, every tree
with $m-1$ edges can be decomposed into two subgraphs
$\mathcal{R}$ and $\mathcal{P}$ such that $\mathcal{R}$ is
1-regular and each component of $\mathcal{P}$ is an edge or a locally irregular component. Let $T$ be a tree with $m$ edges.
Choose an arbitrary vertex $z$ of $T$, and perform a breadth-first search
algorithm from the vertex $z$. This defines a partition $V_0, V_1, \ldots, V_d$ of the vertices of $T$
where each part $V_i$ contains the vertices of $T$ which are at distance exactly $i$
from $z$. Assume that $v \in V_d$, $u \in V_{d-1}$ and $vu\in E(T)$.
Let $T'= T \setminus \{vu\}$, by the inductive hypothesis,  $T'$ can be decomposed into two subgraphs
$\mathcal{R'}$ and $\mathcal{P'}$ such that $\mathcal{R'}$ is
1-regular and each component of $\mathcal{P'}$ is an edge or a locally irregular component. Without loss of generality
 suppose that $w\in V_{d-2}$ and
$uw \in E(T)$. Three cases can be considered:\\
Case 1. If $d_{\mathcal{R'}}(u)=0$. Put $\mathcal{R}= \mathcal{R'} \cup \{uv \}$. In this case $(\mathcal{P}=\mathcal{P'}, \mathcal{R}=\mathcal{R'} \cup \{uv \})$
is a suitable partition for $T$.\\
Case 2. If $d_{\mathcal{R'}}(u)=1$ and $uw \in \mathcal{R'}$. Put $\mathcal{P}= \mathcal{P'} \cup \{uv \}$.
It is easy to see that $(\mathcal{P}=\mathcal{P'} \cup \{uv \}, \mathcal{R}=\mathcal{R'})$
is a suitable partition for $T$.\\
Case 3. If $d_{\mathcal{R'}}(u)=1$ and $uw \notin \mathcal{R'}$. Without loss of generality, assume that $e$ is incident with the vertex $u$ and $e\in \mathcal{R'}$.
Let $\mathcal{P''}= \mathcal{P'} \cup \{uv \}$.
 According to the degrees of $u$ and $w$ in $\mathcal{P'}$, one of the two partitions
 $(\mathcal{P''}, \mathcal{R'})$ or $(\mathcal{P''}\cup \{e\}, \mathcal{R'}\setminus \{e\})$
is a suitable partition for $T$. This completes the proof.\\ \\
(iii) Let $T$ be a tree, by (ii), $T$ can be decomposed into two subgraphs
$\mathcal{P}$ and $\mathcal{R}$ such that $\mathcal{R}$ is
a matching and each component of $\mathcal{P}$ is an edge or a locally irregular component. We can decompose
$\mathcal{P}$ into a matching and a locally irregular subgraph, thus $T$ can be decomposed into two
matchings and a locally irregular subgraph. This completes the proof.
}\end{alii}

In \cite{bensmail2013complexity}, Baudon {\it et al.} proved that determining whether a given planar graph $G$,
can be decomposed into two
locally irregular subgraphs is {\bf NP}-complete.
But their reduction does not preserve the planarity. We show
 that determining whether a given planar graph $G$,
can be decomposed into two locally irregular subgraphs is {\bf NP}-complete by a different reduction.

\begin{aliii}{
We reduce {\em  Monotone Planar 2-In-4 4-Sat} to our problem.
 Kara {\it et al.} \cite{kara2007complexity} proved
that the following problem is $ \mathbf{NP}$-complete.

 {\em  Monotone Planar 2-In-4 4-Sat.}\\
\textsc{Instance}: A 4-Sat formula $\Phi=(X,C)$
 such that there
 is no negation in the formula, and the
bipartite graph obtained by linking a variable and a clause if and only
 if the
 variable appears in the clause, is planar.\\
\textsc{Question}: Is there a truth assignment for $X$ such that
 each clause in $C$ has exactly
two true literals?

Let $\Psi=(X,C)$ be an instance of {\em Monotone Planar 2-In-4 4-Sat}. We denote the number of clauses containing the variable $x$ by
$\gamma(x)$.
We convert $\Psi$ into a planar
graph $G_\Psi$ such that $\Psi$ has a 2-in-4 satisfying assignment if and
only if  the edge set of $G_\Psi$ can be decomposed into two locally irregular subgraphs.
First, consider the gadget $\mathcal{K}$ which is shown in Fig \ref{p1}. We will use  the gadget $\mathcal{K}$ several times and it has some important properties.

\begin{lem}\label{lem001}
Suppose that $G$ is a graph and has a copy of $\mathcal{K}$ as an induced subgraph. Also, assume that the graph $G$ can be decomposed
into two locally irregular subgraphs, then (up to symmetry) the set of black edges  (see Fig \ref{p1}) is in one part and the set
of blue edges is in another part.
\end{lem}

\begin{proof}{
By considering all possible cases, the proof is straightforward.
}\end{proof}

We will use the gadget $\mathcal{K}$ in order to construct our main gadgets in our reduction.

\begin{figure}[ht]
\begin{center}
\includegraphics[scale=.55]{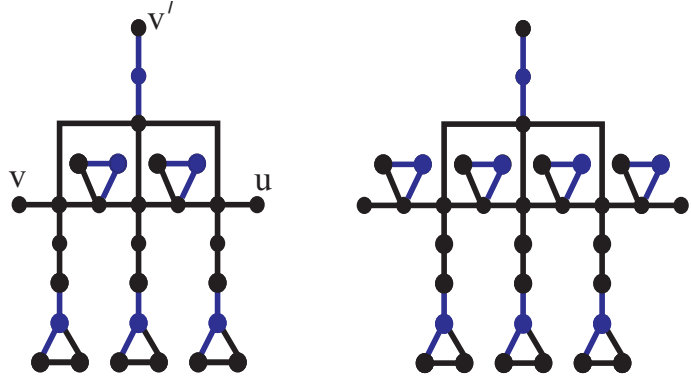}
\caption{The two auxiliary gadgets $\mathcal{L}(v,u)$ and $ \mathcal{K}$. $\mathcal{K}$ is on the right side of the figure.
} \label{p1}
\end{center}
\end{figure}

{\bf Construction of the gadget $\mathcal{A}_{\gamma(x)}(x)$.}\\
Consider a cycle of length $\gamma(x)$ with the vertices
$x_0,x_1, \ldots, x_{\gamma(x)-1}$, in that order. For each $i$, $i=0,1,\ldots, \gamma(x)-1$, replace the edge
 $x_ix_{(i+1 \mod \gamma(x))}$ of the cycle by a copy of
the gadget $\mathcal{L}(x_i,x_{(i+1 \mod \gamma(x))})$ ($\mathcal{L}(v,u)$ is shown in Fig \ref{p1}).
Finally, for each $i$, $i=0,1,\ldots, \gamma(x)-1$, put two new vertices $z_i$, $w_i$ and join the
vertex $x_i$ to the vertices $z_i$, $w_i$. Also, join the vertex $z_i$ to the vertex $w_i$. Call
the resultant auxiliary graph $\mathcal{A}_{\gamma(x)}(x)$.

There are exactly $\gamma(x)$
vertices $x'_0,x'_1, \ldots, x'_{\gamma(x)-1}$
of degree one
in $\mathcal{A}_{\gamma(x)}(x)$. We will call these vertices the
{\it important vertices} of $\mathcal{A}_{\gamma(x)}(x)$ (see Fig \ref{p1}).
Suppose that $G$ is a graph and has a copy of $\mathcal{A}_{\gamma(x)}(x)$ as an induced subgraph.
If the graph $G$ can be decomposed
into two locally irregular subgraphs, then the set of edges incident with the important vertices is in a same
part ({\bf Fact 1}).

{\bf Construction of the clause gadget $\mathcal{B}_c$.}\\
Consider a copy of the gadget $\mathcal{A}_{18}(v)$; also, add five vertices
$c, \alpha_c, \beta_c, \gamma_c, \zeta_c$ and the set of edges
$\{c\alpha_c, c\beta_c, c\gamma_c, c\zeta_c \}$.
Next, consider 18 paths of lengths
2,2,4,2,2,2,2,2,4,2,4,4,2,4,4,4,4,4
 and call them  $P_0, P_1, \ldots, P_{17}$,
respectively.
Identify one of the ends  of $P_0,P_1,P_2$ with $\alpha_c$. Identify one of the ends  of $P_3,P_4, \ldots, P_8$ with $\beta_c$.
Identify one of the ends  of $P_9,P_{10}, P_{11}$ with $\gamma_c$ and
identify one of the ends  of $P_{12}, \ldots, P_{17}$ with $\zeta_c$.
For each $i$, $i=0,1,\dots,17$, identify the other end of path $P_i$ with the important vertex $v'_i$
of $\mathcal{A}_{18}(v)$. Call the resultant gadget $\mathcal{B}_c$. See Fig. \ref{p0002}.

\begin{figure}[ht]
\begin{center}
\includegraphics[scale=.51]{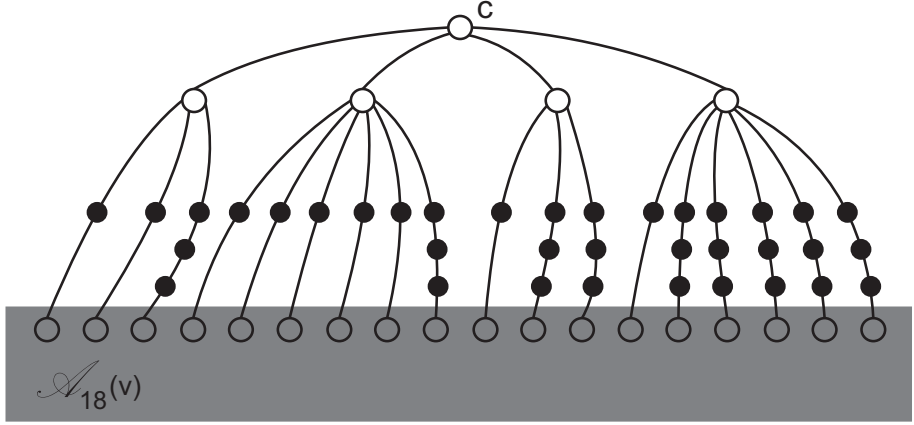}
\caption{The gadget $\mathcal{B}_c$.
} \label{p0002}
\end{center}
\end{figure}

By the structure of the gadget $\mathcal{B}_c$, if the graph $G$ can be decomposed
into two locally irregular subgraphs $\mathcal{I}_1$, $\mathcal{I}_2$,
then exactly two of the edges $c\alpha_c, c\beta_c, c\gamma_c, c\zeta_c$ are in $\mathcal{I}_1$ ({\bf Fact 2}).
Note that the vertices $\alpha_c,\beta_c, \gamma_c, \zeta_c$ are incident with 4,7,4,7 edges respectively and
exactly 3,6,1,1 of them
are in $\mathcal{I}_1$ or vice-versa respectively ({\bf Fact 3}).

Now, we are ready to define the graph $G_\Psi$. For every variable $x\in X$, put a copy of the gadget $\mathcal{A}_{\gamma(x)}(x)$ and
for each clause $c\in C$, put a copy of the gadget $\mathcal{B}_c$. For every pair $(x,c)$, if $x$ appears
in $c$, then put an edge between the vertex $c$ of the gadget $\mathcal{B}_c$ and one of the important
vertices $x'_0,x'_1, \ldots, x'_{\gamma(x)-1}$ of the gadget $\mathcal{A}_{\gamma(x)}(x)$, such that
having done this procedure for all pairs, the degree of each important vertex is two. Call the resultant planar graph $G_\Psi$.
Let $c=(x \vee y \vee z \vee w)$ be an arbitrary clause and without loss of generality suppose
that $c x'_0, c y'_0, c z'_0,c w'_0  \in E(G_\Psi) $. By Fact 2 and Fact 3,
if the graph $G_\Psi$ can be decomposed
into two locally irregular subgraphs $\mathcal{I}_1$, $\mathcal{I}_2$,
then exactly two of the edges $c x'_0, c y'_0, c z'_0,c w'_0$ are in $\mathcal{I}_1$. Thus, we can
find a 2-in-4 satisfying assignment. On the other hand, assume that the formula $\Psi$ has a 2-in-4 satisfying assignment $\Gamma$. For a given clause $c=(x \vee y \vee z \vee w)$, without loss of generality suppose
that $c x'_0, c y'_0, c z'_0,c w'_0  \in E(G_\Psi) $. Then for each literal $v$ ($v\in \{x,y,z,w\}$) put $c v'_0$ in $\mathcal{I}_1$ if and only if $\Gamma(v)=true$. One can extend this to a proper decomposition.
This completes the proof.

}\end{aliii}

It was proved that computation of the regular number is {\bf NP}-hard for connected bipartite graphs \cite{reg}.
Also, it was shown that deciding whether $reg(G) = 2$ for a given connected 3-colorable
graph $G$ is {\bf NP}-complete \cite{reg}. We improve these two results, and show that for a given
bipartite graph $G$ with maximum degree six, deciding whether $reg(G) = 2$ is {\bf NP}-complete. Furthermore,
we present a polynomial time algorithm to decide whether $ reg(G) = 2 $ for a given graph $ G $
with maximum degree five.

\begin{aliiii}{
(i)
It has been shown that the following version of
 {\em Not-All-Equal (NAE) satisfying assignment problem} is {\bf NP}-complete \cite{reg}.

 {\em Cubic Monotone  NAE (2,3)-Sat.}\\
\textsc{Instance}: Set $X$ of variables, collection $C$ of clauses over $X$ such that each
clause $c \in C$ has $\mid c  \mid \in \{ 2, 3\}$, every variable appears in
exactly three clauses and there is no negation in the formula.\\
\textsc{Question}: Is there a truth assignment for $X$ such that each clause in $C$ has at
least one true literal and at least one false literal?\\

Let $\alpha\geq 3$ be a fixed number. We reduce {\em Cubic Monotone NAE (2,3)-Sat } to
our problem in polynomial time.
Consider an instance $ \Phi $, we transform this into a bipartite  graph $G_{\Phi}$ in polynomial time
such that $ reg(G_{\Phi})=2$ if and only if $\Phi$ has an  NAE truth assignment.
We use two auxiliary gadgets $ \mathcal{H}^{\alpha}_c$ and $\mathcal{I}^{\alpha}_c$. See Fig. \ref{p0001}.

\begin{figure}[ht]
\begin{center}
\includegraphics[scale=.5]{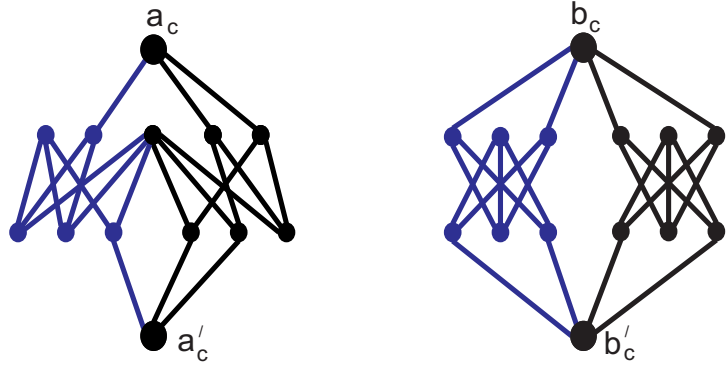}
\caption{The two auxiliary gadgets $\mathcal{H}^{3}_c$ and $\mathcal{I}^{3}_c$. The gadget $\mathcal{H}^{3}_c$ is one the left side.
} \label{p0001}
\end{center}
\end{figure}

{\bf Construction of the gadget $\mathcal{I}^{\alpha}_c$.}\\
 Consider two copies of the complete bipartite graph $K_{\alpha, \alpha}$ and call them $K[X,Y]$, $K'[X',Y']$, where
$X=\{x_i :  1 \leq i \leq \alpha \}, X'=\{x'_i: 1 \leq i \leq \alpha \}, Y=\{y_i: 1 \leq i \leq \alpha\}, Y'=\{y'_i: 1 \leq i \leq \alpha\}$.
Add two vertices $b_c$ and $b_c'$. Join the vertex $b_c$ to the vertices $x_1,\ldots, x_{\alpha-1},
x'_1,\ldots, x'_{\alpha-1}$
and join the vertex $b_c'$ to the vertices $y_1,\ldots, y_{\alpha-1},y'_1,\ldots, y'_{\alpha-1}$.
Finally, remove a perfect matching between the two sets of vertices
 $x_1,\ldots, x_{\alpha-1},x'_1,\ldots, x'_{\alpha-1}$ and
$ y_1,\ldots, y_{\alpha-1},y'_1,\ldots, y'_{\alpha-1}$. Call the resulting gadget $\mathcal{I}^{\alpha}_c$. See Fig. \ref{p0001}.

Note that in the gadget $\mathcal{I}^{\alpha}_c$ the degrees of all vertices except $b_c$ and $b_c'$ are $\alpha$. Also, the degrees of the vertices $b_c$ and $b_c'$ are $2\alpha-2$.

{\bf Construction of the gadget $\mathcal{H}^{\alpha}_c$.}\\
Consider two copies of the complete bipartite graph $K_{\alpha, \alpha}$ and call them $K[X,Y]$, $K'[X',Y']$, where
$X=\{x_i: 1 \leq i \leq \alpha\}, X'=\{x'_i: 1 \leq i \leq \alpha\}, Y=\{y_i: 1 \leq i \leq \alpha\}, Y'=\{y'_i: 1 \leq i \leq \alpha\}$.
Add two vertices $a_c$ and $a_c'$. Identify the vertex $x_{\alpha}$ with the vertex $x'_{\alpha}$.
Join the vertex $a_c$ to the vertices $x_1,\ldots, x_{\alpha-1}, x'_1,\ldots, x'_{\alpha-2}$
and join the vertex $a_c'$ to the vertices $y_1,\ldots, y_{\alpha-1},y'_1,\ldots, y'_{\alpha-2}$.
Finally, remove a perfect matching between the set of
vertices $x_1,\ldots, x_{\alpha-1},x'_1,\ldots, x'_{\alpha-2} $ and the set of vertices
$y_1,\ldots, y_{\alpha-1},y'_1,\ldots, y'_{\alpha-2}$. Call the resulting gadget $\mathcal{H}^{\alpha}_c$.
See Fig. \ref{p0001}.

Note that in the gadget $\mathcal{H}^{\alpha}_c$ the degrees of all vertices except $a_c$ and $a_c'$ are $\alpha$.
Also, the degrees of the vertices $a_c$ and $a_c'$ are $2\alpha -3$.
The graph $G_{\Phi}$ has
a copy of the gadget $\mathcal{H}^{\alpha}_c $ for each clause $c \in C$ with $\mid c \mid =3$, and a copy
of the gadget $\mathcal{I}^{\alpha}_c $ for each clause $c \in C$ with $\mid c \mid =2$.
Also, for each variable $x \in X$, put two  vertices $x$ and $x'$ and
consider a copy of the complete bipartite graph $K_{\alpha, \alpha}$ and call it $K[X,Y]$, where
$X=\{x_i: 1 \leq i \leq \alpha\}, Y=\{y_i: 1 \leq i \leq \alpha\}$.
Join the vertex $x$ to the vertices $x_1,\ldots, x_{\alpha-3}$
and join the vertex $x'$ to the vertices $ y_1,\ldots, y_{\alpha-3}$.
Next, remove a perfect matching between the set of
vertices $x_1,\ldots, x_{\alpha-3}$ and the set of vertices
$y_1,\ldots, y_{\alpha-3}$.
Finally, for each clause $c =(y \vee z \vee w)$, where  $y,w,z \in X  $ add the edges $a_c y $,
$a'_c y'$, $a_c z  $, $a'_c z'$, $a_c w $ and $a'_c w'$. Also,
for each clause $c =(y \vee z )$, where  $y,w \in X  $ add the
edges $b_c y $, $b'_c y' $, $b_c z$ and $b'_c z'  $.
Call the resultant graph $G_{\Phi}$. The degree of every vertex in $G_{\Phi}$ is $\alpha$ or $2\alpha$ and the graph is bipartite.
Note that there are no two adjacent vertices of degree $2\alpha$.

First, assume that $reg(G_{\Phi})=2$ and let  $G_1 $ and $G_2$ be a regular decomposition of the graph $G_{\Phi}$ such that $G_i$
is $(r_i)$-regular, for each $i$, $i=1,2$.
The graph $G_{\Phi}$ has vertices with degrees $\alpha$ and $2\alpha$, so $r_1=r_2=\alpha$.
For every $x$, $x \in X$,  the vertex $x$ has degree $\alpha$, so all edges incident with the vertex
 $x$ are in the same part. For every $x\in X$,
if all edges incident with the vertex $x$ are in $G_1$, put $\Gamma(x)=true$ and if all edges incident with
the vertex $x$ are in $G_2$,
put $\Gamma(x)=false$.
According to the construction of $\mathcal{H}^{\alpha}_c$, ($\mathcal{I}^{\alpha}_c$, respectively),
 the set of edges $\{a_cx_1,\ldots, a_c x_{\alpha-1} \}$ ($\{b_cx_1,\ldots, b_c x_{\alpha-1} \}$, respectively)
 is in one part and the set  of  edges
 $\{a_cx'_1,\ldots, a_c x'_{\alpha-2} \}$ ($\{b_cx'_1,\ldots, b_c x'_{\alpha-1}\}$, respectively) is in another part.
Therefore, for every clause $c=(y \vee z \vee w)$, at most two of the three edges $a_c y $, $a_c z  $
and $a_c w $ are in $G_1$. Also, at most two of the three edges $a_c y $, $a_c z  $ and $a_c w $ are in $G_2$.
Similarly, for every clause $c=(y \vee z) $, exactly one of the two edges $b_c y $ and $b_c z  $ is in $G_1$
(Note that for every clause $c=(y \vee z \vee w)$, at most two of the three edges $a'_c y' $, $a'_c z'  $
and $a'_c w' $ are in $G_1$ and similarly exactly one of the two edges $b'_c y' $ and $b'_c z'$ is in $G_1$).
Hence, $\Gamma$ is an NAE satisfying assignment.
On the other hand, suppose that the formula $\Phi$ has an NAE satisfying assignment  $\Gamma : X  \rightarrow \{true, false\}$. For every variable $x\in X$, put all
edges incident with the the vertices $x$ and $x'$ in $G_1$ if and only if $\Gamma (x)=true$. By this method,
it is easy to show that $G_{\Phi}$ can be decomposed into two regular subgraphs.
This completes the proof.
\\
\\
(ii) For a given connected graph $G$, assume that $\Delta(G)=5$ and
let  $G_1 $ and $G_2$ be a regular decomposition of the graph $G$ such that $G_i$ is $(r_i)$-regular, for each $i$, $i=1,2$.
If $G$ is not a 5-regular graph, then two cases can be considered:

Case 1.1: $r_1=2$ and $r_2=3$. In this case the degree set of the graph must be a subset of $\{2,3,5\}$. Let
$G'$ be the induced graph on the set of vertices of degrees two and five.
A subgraph $F$ of a graph $G'$ is called a factor of $G'$ if $F$ is a
spanning subgraph of $G'$. If a factor $F$ has all of its degrees equal to $k$, it is called a $k$-factor.
Thus a 2-factor is a disjoint union of finitely many cycles that cover all vertices of $G'$.
It is well-known that the problem of finding a 2-factor can be solved in polynomial time by matching techniques.
In other words, a 2-factor can be constructed in polynomial time if the answer is YES (see \cite{MR1367739}, page 141).
The graph $G$ can be decomposed into a 2-regular and  a 3-regular graphs if and only if the graph  $G'$ has a 2-factor.

 Case 1.2: $r_1=1$ and $r_2=4$. In this case the degree set of graph must be a subset of $\{1,4,5\}$. Let
$G'$ be the induced graph on the set of vertices of degrees one and five.
It is easy to see that the graph $G$ can be decomposed into a 1-regular and a 4-regular graphs
 if and only if the graph $G'$ has a perfect matching.
Finding the maximum matching is in $ \mathbf{ P} $ (see \cite{MR1367739}, page 145), hence, the proof is completed.

Now, assume that $\Delta(G)=4$. If the graph $G$ is not a 4-regular graph, then two cases can be considered:

 Case 2.1: $r_1=2$ and $r_2=2$. In this case the degree  set of graph must
be a subset of $\{2,4\}$.
Consider the graph $H$ with the vertex set $\{v| d(v)=4, v\in V(G)\}$ and join two vertices $v$ and $u$ in $H$ if and only if there
is a path $v x_1\ldots x_t u$ in $G$ such that $d(x_1)=\cdots=d(x_t)=2$ (note that if $vu\in E(G)$, then we join
the vertex $v$ to the vertex $u$).
A $k$-factorization of $H$ is a partition of the
edges of $H$ into disjoint $k$-factors.
For $k\geq 1$, every $2k$-regular graph admits a $2$-factorization (see \cite{MR1367739}, page 140), therefore $H$ has a 2-factor,
thus the graph $G$ has a 2-factor. Consequently, in this case the graph $G$ always can be decomposed into
two regular subgraphs.

Case 2.2: $r_1=1$ and $r_2=3$. In this case the degree set of graph must be a subset of $\{1,3,4\}$. Let
$G'$ be the induced graph on the set of vertices of degree one and four.
It is easy to see that the graph  $G$ can be decomposed into a 1-regular and a 3-regular subgraphs
if and only if the graph $G'$ has a perfect matching.
Finding the maximum matching is in $ \mathbf{ P} $ (see \cite{MR1367739}, page 145), consequently the proof is completed.

Now, assume that $\Delta(G)=3$. If $G$ is not a 3-regular graph, then $r_1=1$ and $r_2=2$. Let
$G'$ be the induced graph on the set of vertices of degrees one and three.
It is easy to see that the graph $G$ can be decomposed into a 1-regular and a 2-regular graphs if and only if the graph $G'$ has a perfect matching.
The other cases for $\Delta(G)=1$ and $\Delta(G)=2$ are trivial. This completes the proof.

}\end{aliiii}

Note that in the proof of previous theorem from Case 2.1,
we have the following corollary.

\begin{cor}{
Let $G$ be graph with degree  set $\{2,4\}$. The edge set of the graph $G$ can be decomposed into two subgraphs
such that each subgraph is 2-regular.
}\end{cor}

Next, we consider the problem of determining the regular number for planar graphs.

\begin{alij}{
 Clearly, the problem is in $ \mathbf{NP} $. We reduce {\em  Monotone Planar 2-In-4 4-Sat} to our problem.
 Kara {\it et al.} \cite{kara2007complexity} proved
that the following problem is $ \mathbf{NP}$-complete.

 {\em  Monotone Planar 2-In-4 4-Sat.}\\
\textsc{Instance}: A 4-Sat formula $\Psi=(X,C)$
 such that there
 is no negation in the formula, and the
bipartite graph obtained by linking a variable and a clause if and only
 if the
 variable appears in the clause, is planar.\\
\textsc{Question}: Is there a truth assignment for $X$ such that
 each clause in $C$ has exactly
two true literals?

Let $\Psi=(X,C)$ be an instance of {\em Monotone Planar 2-In-4 4-Sat}.
We denote the number of the clauses containing the variable $x$ by
$\gamma(x)$.
We convert $\Psi$ into a planar
graph $G_\Psi$ with degree set $\{3,6\}$ such that $\Psi$ has a 2-in-4 satisfying assignment if and
only if  the edge set of $G_\Psi$ can be decomposed into two  regular subgraphs.
For each variable $x$, consider a cycle of length $\gamma(x)$ with the vertices $x_1, \ldots, x_{\gamma(x)}$,
in that order
and for each clause $c$ consider a copy the gadget $\mathcal{Z}_c$ which is shown in Fig \ref{P2}.
Finally, for every pair $(x,c)$, where $x\in X$ and $c\in C$, if $x$ appears in $c$, join the vertex $a_c$
to one of the vertices $x_1,x_2, \ldots, x_{\gamma(x)}$, such that in the resulting graph the degree of
every vertex in $\{x_1,x_2, \ldots, x_{\gamma(x)}\}$ is three. Call the resultant planar graph $G_\Psi$.

\begin{figure}[ht]
\begin{center}
\includegraphics[scale=.5]{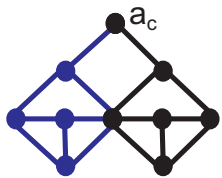}
\caption{The auxiliary gadget  $\mathcal{Z}_c$.
} \label{P2}
\end{center}
\end{figure}

First, assume that $reg(G_\Psi)=2$ and let  $G_1 $ and $G_2$ be a regular decomposition of the graph $G_\Psi$ such that $G_i$
is $(r_i)$-regular, for each $i$, $i=1,2$.
The graph $G_\Psi$ has vertices with degrees $3$ and $6$, so $r_1=r_2=3$.
For every $x \in X$, the  degrees of vertices $x_1,x_2, \ldots, x_{\gamma(x)}$ are $3$, thus, all edges incident
with one of the vertices $x_1,x_2, \ldots, x_{\gamma(x)}$ are in the same part. For every $x\in X$,
if all edges incident with the vertex $x_1$ are in $G_1$, put $\Gamma(x)=true$ and if all edges incident with
the vertex $x_1$ are in $G_2$,
put $\Gamma(x)=false$.
According to the construction of $\mathcal{Z}_c$
 the set of black edges are in one part and the set  of blue edges
 are in another part.
Therefore, for every clause $c=(x \vee y \vee z \vee w)$, exactly two edges from the four edges $a_c x $, $a_c y $, $a_c z  $
and $a_c w $ are in the subgraph $G_1$.
Hence, the function $\Gamma$ is a 2-in-4 satisfying assignment.
On the other hand, suppose that the formula $\Phi$ is 2-in-4 satisfiable with the satisfying
assignment $\Gamma : X  \rightarrow \{true, false\}$. For every variable $x\in X$, put all
edges incident with the vertices $x_1,x_2, \ldots, x_{\gamma(x)}$ in $G_1$ if and only if $\Gamma (x)=true$,
also for every gadget $\mathcal{Z}_c$, put the black edges in $G_1$ and the blue edges in $G_2$.
By this method,
it is easy to show that the graph $G_\Psi$ can be decomposed into two regular subgraphs.
This completes the proof.

}\end{alij}

It was shown that determining whether $reg(G)\leq \Delta(G)$ for a given
connected graph $ G $  is {\bf NP}-complete \cite{reg}. Here, we show that every graph $G$ can be decomposed into $\Delta(G)$
subgraphs such that each subgraph is locally regular and this bound is sharp for trees.

\begin{alijj}{
We use the concept of semi-coloring to prove our
theorem. Daniely and Linial defined a semi-coloring of graphs for the investigation of the
tight product of graphs \cite{daniely2012tight}. Afterwards, Furuya {\it et al.} proved that every graph has
a semi-coloring \cite{furuya2014existence}. Let  $G$ be a graph. For $i\in \{1,2\}$, let $ {[\Delta(G)]} \choose i$ denotes
the family of subsets of $\{1,2,\ldots, \Delta(G)\}$ with cardinality exactly $i$. A semi-coloring of the graph $G$ is a coloring

\begin{center}
$\ell:  E(G) \rightarrow {{[\Delta(G)]} \choose 1} \cup {{[\Delta(G)]} \choose 2}$
\end{center}
\ \\
such that for every $v \in V (G)$,
\\
(1) For each $i$, $i=1,2,\ldots, \Delta(G)$, $\sum_{e \ni v} w_i(e) \in {0,1}$, where

\begin{center}{

$ w_i(e)=\left\{ \begin{array}{ll}  \frac{1}{|\ell(e)|} &   i\in \ell(e) \\ 0 & i\notin \ell(e) \ .
 \end{array} \right.$

}\end{center}
\ \\
and
\\
(2) For any $1 \leq i < j \leq \Delta(G)$, $|\{e: e \ni v , \ell(e)=\{i,j\} \}| \in \{0,2\}$.

Let $G$ be a graph
with maximum degree $\Delta(G)$ and semi-coloring  $\ell$. Define the following decomposition for the edges of the graph $G$.
\\
\\
For each $i$, $i=1,2,\ldots, \Delta(G)$, $\mathcal{P}_i=\{ e : \ell(e)=\{i\}  \} \cup \{  e: \ell(e)=\{i,j\}, i <j \}$.

For any semi-coloring  $\ell$ of the graph $G$ and any $i,j$, where $1 \leq i < j \leq \Delta(G) $, each component
of the subgraph of $G$ induced by edges with the color $\{i, j\}$ is a singleton or a cycle. So, each component
of each part $\mathcal{P}_i$ is an edge or a cycle. Thus, the regular chromatic index of the graph $G$ is at most $\Delta(G)$.
Since every graph has a semi-coloring \cite{furuya2014existence}, so our proof is completed.

Now, let $T$ be a tree. The tree $T$ does not have any cycle, so in every edge decomposition, each part is a matching.
On the other hand, the edge chromatic number of the tree $T$ is equal to  $ \Delta(T) $ \cite{MR1367739}, therefore,
 $\chi_{reg}'(T)= \Delta(T)$. This completes the proof.
}\end{alijj}

In Theorem \ref{T3}, we prove that there is polynomial time algorithm to decide whether $ reg(G) = 2 $ for a given graph $ G $
with maximum degree five. Here, we show that deciding whether a given  subcubic graph can be decomposed into two
subgraphs such that each subgraph is locally regular, is {\bf NP}-complete.

\begin{alijjj}{
We reduce {\em Cubic Monotone  NAE (2,3)-Sat} to our problem.
It is shown that the following version of  {\em NAE satisfying assignment problem} is {\bf NP}-complete \cite{reg}.

{\em Cubic Monotone  NAE (2,3)-Sat.}\\
\textsc{Instance}: Set $X$ of variables, collection $C$ of clauses over $X$ such that each
clause $c \in C$ has $\mid c  \mid \in \{ 2, 3\}$, every variable appears in
exactly three clauses and there is no negation in the formula.\\
\textsc{Question}: Is there a truth assignment for $X$ such that each clause in $C$ has at
least one true literal and at least one false literal?\\

Let $\Phi=(X,C)$ be an instance of {\em  Cubic Monotone  NAE (2,3)-Sat}.
We convert the formula $\Phi$ into a
graph $G$ such that $\Phi$ has an NAE satisfying assignment if and
only if  the edge set of $G$ can be decomposed into two
subgraphs such that each subgraph is locally regular.
We use two gadgets $\mathcal{U}(c)$ and $\mathcal{W}(x)$. The gadget $\mathcal{U}(c)$ is shown in Fig. \ref{P9}.

\begin{figure}[ht]
\begin{center}
\includegraphics[scale=.57]{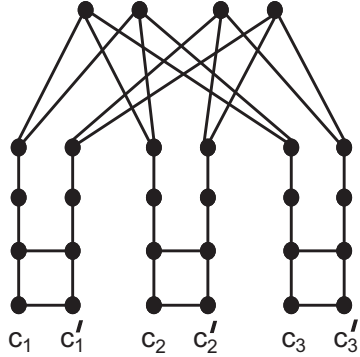}
\caption{The auxiliary gadget  $\mathcal{U}(c)$.
} \label{P9}
\end{center}
\end{figure}

{\bf Construction of $\mathcal{W}(x)$}\\
Consider a cycle $C_6$, with vertices $v_1, u_1, v_2,u_2, v_3,u_3$, in that order. For each  $i$, $1 \leq i \leq 3$,
put a triangle and join one the vertices of that triangle to the vertex $u_i$. The resultant graph  has nine vertices with degree two
and six vertices with degree three. The three edges between triangles and $C_6$ are called {\it blue edges} and other edges of that graph
are called {\it black edges}. Call the resultant graph $\mathcal{S}(x)$. Now consider two copies of $\mathcal{S}(x)$ and rename the
vertices $v_1,v_2,v_3$ ($v_1,v_2,v_3$) in the first copy (second copy) of $\mathcal{S}(x)$, by $x_1,x_2,x_3$ ($x'_1,x'_2,x'_3$),
respectively.
Call the resulting gadget $\mathcal{W}(x)$.

{\bf Construction of the graph $G$}\\
Let $\Phi=(X,C)$ be an instance of {\em  Cubic Monotone  NAE (2,3)-Sat}.
For every clause $c \in C$, if $|c|=3$, then put a copy of the  gadget $\mathcal{U}(c)$ and if $|c|=2$ put two vertices $c$ and $c'$.
Also, for every variable $x \in X$, put a copy of the gadget $\mathcal{W}(x)$.
Now, for any clause $c$ containing $x$,
choose an index $i$, $1\leq i \leq 3$, if $|c|=2$, then put the two edges $cx_i$, $c'x_i'$ and if $|c|=3$
choose an index $j$, $1\leq j \leq 3$ and add the two edges $c_jx_i$, $c_j'x_i'$.
Do these so that in the resultant graph $G$,
the degree of every vertex is at most 3.

Suppose that the graph $G$ can be decomposed into two
subgraphs such that each subgraph is locally regular. Since $G$ is subcubic,  each component of every subgraph is an edge or a
cycle. Thus, in each copy of the gadget $\mathcal{W}(x)$, the set of blue edges is in one subgraph
and the set of black edges is in another
subgraph ({\bf Property 1}). For each clause $c=(x,y)$, without loss of generality suppose that $cx_1, cy_1 \in E(G)$.
By the structure of $G$, the two edges $cx_1 $ and $ cy_1$ are in different subgraphs, similarly,
for every clause $c=(x,y,z)$, without loss of generality suppose that $c_1x_1, c_2 y_1, c_3 z_1 \in E(G)$.
By the structure of $\mathcal{U}(c)$, the
edges $c_1x_1, c_2 y_1, c_3 z_1$ are not in the same subgraph ({\bf Property 2}).
Now, assume that $G$ can be decomposed into two
subgraphs $G_1$ and $G_2$ such that each subgraph is locally regular.
For every $x\in X$, if the two edges incident with $x_1$ in the gadget $\mathcal{W}(x)$ are in $G_1$, put $\Gamma(x)=true$ and
otherwise put $\Gamma(x)=false$. By Property 1 and Property 2, it is easy to check that $\Gamma(x)$ is an NAE assignment.

}\end{alijjj}

Here, We prove that for
each $k>1 $, deciding whether $\chi'_{k-irr}(G)=2$ for a given planar bipartite graph $G$  is ${\bf NP}$-complete.
For all $k$ we prove the lower bound $h(k) \geq 2k+1$ and we will use
 mutually orthogonal Latin squares and prove that  $h(k)=\Omega(k^2)$.

\begin{alijjjk}{
{\bf (i)}
We reduce {\em  Monotone Planar 2-In-4 4-Sat} to our problem.
 Kara {\it et al.} \cite{kara2007complexity} proved
that the following problem is $ \mathbf{NP}$-complete.

 {\em  Monotone Planar 2-In-4 4-Sat.}\\
\textsc{Instance}: A 4-Sat formula $\Phi=(X,C)$
 such that there
 is no negation in the formula, and the
bipartite graph obtained by linking a variable and a clause if and only
 if the
 variable appears in the clause, is planar.\\
\textsc{Question}: Is there a truth assignment for $X$ such that
 each clause in $C$ has exactly
two true literals?

Assume that $k\geq 2$ is a fixed integer. Let $\Phi=(X,C)$ be an instance of {\em  Monotone Planar 2-In-4 4-Sat}.
 We convert the formula $\Phi$ into a
 graph $G_k$ such that the formula $\Phi$ has a 2-in-4 satisfying assignment if and
 only if  the edge set of the graph $G_k$ can be decomposed into two  locally $k$-irregular subgraphs.
First, we introduce two useful gadgets $\mathcal{A}_{\alpha}$ and $\mathcal{B}(c)$.
 We denote the number of the clauses containing the variable $x$ by
$\gamma(x)$.

{\bf Construction of $\mathcal{A}_{\alpha}$}\\
Let $C_{4\alpha}$ be a cycle with $4\alpha$ vertices $v_1, v'_1, u_1, u'_1,  \ldots, v_{\alpha}, v'_{\alpha}, u_{\alpha},
 u'_{\alpha}$ in that order.
For each  $i$, $i=1,2,\ldots, \alpha$, put $2k-2$ new vertices $w^i_1, z^i_1, \ldots, w^i_{k-1}, z^i_{k-1}$ and join the vertex
$v_i$ to the vertices $w^i_1, \ldots, w^i_{k-1}$; also, join the vertex
$u_i$ to the vertices $z^i_1, \ldots, z^i_{k-1}$. Call the resulting bipartite graph $\mathcal{A}_{\alpha}$. In
the gadget $\mathcal{A}_{\alpha}$, call the set of vertices $w^1_1, w^2_1, w^3_1, \ldots, w^{\alpha}_1$, {\it main vertices} and
call the set of edges incident with the main vertices, {\it main edges}.

\begin{figure}[ht]
\begin{center}
\includegraphics[scale=.4]{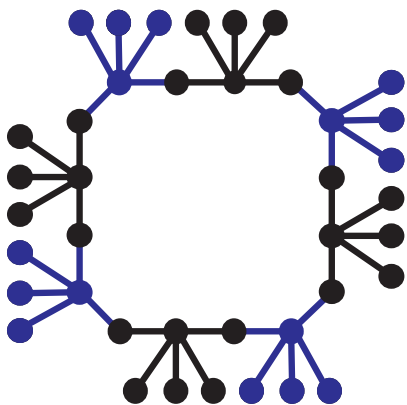}
\caption{The auxiliary gadget  $\mathcal{A}_{4}$, where $k=4$.
} \label{Newpic003}
\end{center}
\end{figure}

{\bf Construction of $\mathcal{B}(c)$}\\
Let $P_5$ be a path with vertices $p_1, p_2, \ldots, p_5$, in that order. Put $2k-2$ new vertices $q_1,q'_1, q_2,q'_2, \ldots,
q_{k-1},q'_{k-1}$, join the vertex $p_2$ to the vertices $q_1,q_2,\ldots, q_{k-1}$ and
join the vertex $p_4$ to the vertices $q'_1,q'_2,\ldots, q'_{k-1}$. Call the resulting graph $\mathcal{D}$.
Note that in $\mathcal{D}$, we have $d(p_1)=d(p_5)=1$. Now put a new vertex $c$ and $k-1$ copies of $\mathcal{D}$.
In each copy of $\mathcal{D}$ join the vertices $p_1$ and $p_5$ to the vertex $c$. Call the resulting graph $\mathcal{B}(c)$.
See Fig. \ref{Newpic003}.

\begin{figure}[ht]
\begin{center}
\includegraphics[scale=.27]{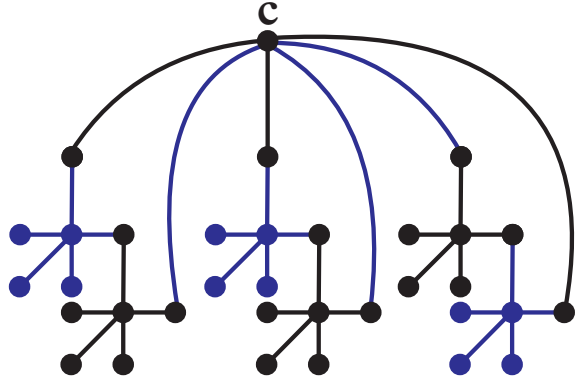}
\caption{The auxiliary gadget  $\mathcal{B}(c)$, where $k=4$.
} \label{Newpic003}
\end{center}
\end{figure}

In the next, we introduce the construction of the graph $G_k$. For every variable $x\in X$, put a copy of the gadget
 $\mathcal{A}_{\gamma(x)}$ (we
call this copy of the gadget $\mathcal{A}_{\gamma(x)}$, the gadget corresponds to the variable $x$) and
for each clause $c \in C$, put a copy of the gadget $\mathcal{B}(c)$.
Now, for any clause $c$ containing $x$, connect one of the main vertices of a copy of $\mathcal{A}_{\gamma(x)}$ corresponding
to the variable $x$  to the
vertex $c$ of $\mathcal{B}(c)$.
Do these procedures for all variables $x$ and all clauses $c$, in such a way that in the resultant graph $G_k$ the degree of every main vertex is
two and for every $c\in C$ the degree of vertex $c$ is $2k+2$.

First, suppose that the graph $G_k$ can be decomposed into two locally $k$-irregular subgraphs $\mathcal{I}$ and $\mathcal{I}'$.
Since the  degree set of $G_k$ is $\{1,2,k+1,2k+2\} $, by the structure of the graph $G_k$, for every vertex $v$ of degree two, if
$e,e' \ni v$, then $e\in E(\mathcal{I})$ and $e\in E(\mathcal{I}')$ or vice versa.
By the structure of $\mathcal{A}_{\alpha}$, for each copy of $\mathcal{A}_{\alpha}$, the set of its main edges is in $\mathcal{I}$
or $\mathcal{I}'$ (Fact 1).
Furthermore, by the structure of $\mathcal{D}$, if we consider the induced graph on the  set of vertices $V(\mathcal{B}(c))$,
we have $d_{\mathcal{I}}(c)=d_{\mathcal{I'}}(c)=k-1$ (Fact 2).

Now, assume that $G_k$ can be decomposed into two locally $k$-irregular subgraphs $\mathcal{I}$, $\mathcal{I}'$.
Let $\Gamma : X \rightarrow \{true,false \} $ be a function such that $\Gamma(x)=true$ if and only if  the set of main edges
of the gadget $\mathcal{A}_{\gamma(x)}$ corresponding to the variable $x$ is in $\mathcal{I}$.
By Fact 1 and 2, it is easy to see that
$\Gamma$ is a 2-in-4 satisfying assignment.
Next, suppose that $\Gamma : X \rightarrow \{true,false \} $ is a 2-in-4 satisfying assignment. For every variable $x$
put the set of main edges
of the gadget $\mathcal{A}_{\gamma(x)}$ corresponds to the variable $x$ in $\mathcal{I}$ if and only if $\Gamma(x)=true$.
It is easy
to see that $\mathcal{I}$ can be extended to a $k$-irregular graph such that $G\setminus
E(\mathcal{I})$ is also a $k$-irregular graph. This
completes the proof.
\\
\\
(ii)
In order to show that $h(k) \geq 2k+1$, it is enough to present a graph $G$ such that  $\chi'_{k-irr}(G)\geq 2k+1$.
Consider a cycle $\mathcal{C}=v_1v_2v_3$. Put $k-1$ new vertices and join them to the vertex $v_1$,
also put  $2k$ new vertices and join them to the vertex $v_2$. Next, put $2k-1$ vertices $u_1, u_2, \ldots, u_{2k-1}$
and join them to the vertex $v_3$. Finally, for every $i$, $i=1,2,\ldots, 2k-1$ put $k$ new vertices and join them to the vertex
$u_i$. Call the resultant graph $G$. It is easy to check that $G$ can be decomposed into locally $k$-irregular graphs.
Suppose that $\chi'_{k-irr}(G)< 2k+1$ and let $E_1,E_2, \ldots, E_t$ be a decomposition  of $E(G)$ such
that $G[E_i]$ is locally $k$-irregular for every $i =1,2,\ldots,t$.

There are vertices of degree one in the neighbors of vertex $u_1$, choose one of these vertices call it $z$.
Let $\mathcal{I}$ be the induced graph on the set of edges $E_1$ and without loss of generality assume that
$d_{\mathcal{I}}(z)=1 $. Since $d_{G}(u_1)=k+1 $, we have $d_{\mathcal{I}}(u_1)=k+1 $. Since $\chi'_{k-irr}(G)< 2k+1$, we have $d_{\mathcal{I}}(v_3)=2k+1 $. Therefore $v_1v_3, v_2v_3 \in E_1$.
It is easy to see that $d_{\mathcal{I}}(v_1)=k+1 $ (otherwise we obtain a contradiction). Hence $v_1v_2 \in E_1$. Therefore
$d_{\mathcal{I}}(v_2)\geq 2 $,
but this is a contradiction. Thus $\chi'_{k-irr}(G)\geq 2k+1$.

Now, assume that $k\geq 4$. Consider a cycle $\mathcal{C}=v_1v_2v_3$. Put $k-2$ new vertices and join them to the vertex $v_1$,
also put  $2k-2$ new vertices $u_1,u_2, \ldots, u_{2k-2}$ and join them to the vertex $v_2$. Next, put $3k-2$ vertices
and join them to the vertex $v_3$. Finally, for every $i$, $i=1,2,\ldots, 2k-3$ (note that $i\neq 2k-2$), put $k$ new
vertices and join them to the vertex
$u_i$. Call the resultant gadget $\mathcal{S}$. Now, consider $4k$ copies of $\mathcal{S}$ and three isolated vertices
$z_1,z_2,z_3$. Next, for every $i$, $i=1,2,3$, join the vertex $z_i$ to the vertex $v_i$ in each copy of $\mathcal{S}$.
Call the resulting graph $G$. It is easy to check that $G$ can be decomposed into locally $k$-irregular subgraphs
and in every decomposition of $G$, all edges of a copy of the gadget $\mathcal{S}$ are in the same subgraph.
Thus, by the structure of the graph $G$, $\chi'_{k-irr}(G)= 4k$.\\ \\
(iii)
Let $k$ be a sufficiently large number and $p$ be a
prime number such that $0.1 k \leq p \leq 0.2 k$.
(Bertrand's postulate states that for any integer $d > 3$,
there always exists at least one prime number $p$ with $d< p < 2d-2$ \cite{dressler1972stronger}).

A Latin square of order $n$ is an $n \times n$ matrix such that  every element of $\{1,2,\ldots,n\}$ occurs exactly once in each row
and each column.
A set of Latin squares is called mutually orthogonal Latin squares (MOLS) if every pair of its element Latin squares
is orthogonal to each other \cite{graham1995handbook}. (Two Latin squares $L_1$ and $L_2$ are orthogonal
if for any $(i, j)$, there exists unique $(k, l)$ such that $L_1(k, l) = i$
and $L_2(k, l) = j$).

Let $\mathcal{L}^1, \ldots, \mathcal{L}^{p-1}$ be a set of  MOLS of order $p$, with elements from $\{1, \ldots, p\}$ (If $p$ is prime,
then there exist $p-1$ MOLSs of order $p$ \cite{graham1995handbook}).
Consider $p$ copies of the complete graph $K_{\lfloor k/2 \rfloor +1}$
and let the vertex set of the $\alpha$th copy of $K_{\lfloor k/2 \rfloor +1}$ be $\{v_i^{\alpha}
: 1 \leq i \leq  \lfloor k/2 \rfloor +1 \}$. For each pair $(\alpha,i)$, where $ 1 \leq \alpha  \leq p $ and
$1 \leq i \leq  \lfloor k/2 \rfloor +1 $, add $\lceil k/2 \rceil +i$ new
vertices  $\{w_{ij}^{\alpha}: 1 \leq j \leq \lceil k/2 \rceil +i \}$ to the graph and join them to the vertex $v_i^{\alpha}$
(so the degree of the vertex $v_i^{\alpha}$ is $k+i$).
Finally, for every $i,j,r$ and $a$ identify the vertices $w_{ij}^1$ and $w_{aj}^r$ if and only if the $(i, j)$ element of
$(r-1)$th Latin square is $a$ (for a Latin square $L$ of order $n$, the element on the $i$th row and the $j$th
column is denoted by $(i, j)$ element of $L$). Call the resultant graph $G_k$.

Now, we show $\chi'_{k-irr} (G_k)\geq \Omega (k^2)$.
By the structure of $K_{\lfloor k/2 \rfloor +1}$s and their incident leaves, all of the  edges
$v_{i}^{\alpha} w_{ij }^{\alpha}$, $  1 \leq j \leq \lceil k/2 \rceil +i $
appear in one part; and the degree of the vertex $w_{ij }^{\alpha}$ in that part should be one {\bf (Property A)}.\\ Also,
by the structure of the graph, for every $\alpha$, the two edges $v_i^{\alpha} w_{ij}^{\alpha}$ and
$v_{i'}^{\alpha} w_{i'j}^{\alpha}$ should appear in different parts {\bf (Property B)}.

By the structure of MOLS, for every $\alpha,\alpha', \beta, \beta'$, $2 \leq \alpha < \alpha'\leq p, 1\leq \beta \leq  \beta' \leq p$ there
 are $i,j$
such that $(i, j)$ element of
$(\alpha-1)$th Latin square is $\beta$ and $(i, j)$ element of
$(\alpha'-1)$th Latin square is $\beta'$. So, $w_{ij}^{1}$ and $w_{\beta j}^{\alpha}$ were merged; also $w_{ij}^{1}$
and $w_{\beta'j}^{\alpha'}$
were merged. Thus, two vertices $w_{\beta j}^{\alpha}$ and $w_{\beta'j}^{\alpha'}$ were merged.
Therefore by Property A, the two edges $v_{\beta}^{\alpha} w_{\beta j}^{\alpha}$
and $v_{\beta '}^{\alpha'} w_{\beta'j}^{\alpha'}$ appear
 in two different parts. By this fact and Property B, we have $\chi'_{k-irr} (G_k)\geq \Omega (k^2)$.
\\ \\
(iv) Let $k$ be a fixed number and $G$ be a graph with maximum degree $k+1$.
If $G$ can be decomposed into two locally $k$-irregular subgraphs, then $G$ meets the following three necessary conditions.
\\
{\bf Condition A.} There are no two adjacent vertices of degrees less than $k+1$.
\\
{\bf Condition B.} There are no two adjacent vertices of degrees $k+1$.
\\
{\bf Condition C.} If $u$ and $v$ are two adjacent vertices and $d(u)=k+1$, then $d(v)\leq 2$.

Suppose that $G$ has the above three conditions. Note that these conditions can be checked in polynomial time.
Let $S=\{v : d_G(v)=k+1\}$, and construct the graph $G^*$ with the vertex set $S$.
For every two distinct vertices $u, v \in S$, join the vertex $u$ to the vertex $v$ in $G^*$, if and only if there
is a vertex $z$ in $G$ such that
$vz,uz\in E(G)$ and $d_G(z)=2$. It is easy to check that $G$ can be decomposed into two locally
$k$-irregular subgraphs if and only if
$G^*$ is bipartite. Since, there is a polynomial-time algorithm
for determining whether a given graph  has a chromatic
number at most 2, therefore, our proof is completed.

}\end{alijjjk}

\section{Concluding remarks and further research}

In this work, we considered the set of problems which is related to decomposition of graphs into regular, locally regular
 and locally irregular
subgraphs and we presented some
polynomial time algorithms, {\bf NP}-completeness results, lower bounds and upper bounds for them.
A summary of results and open problems were shown in Table 1 and Table 2. Here, we present some remarks.

There exist infinitely many trees with irregular chromatic
index three \cite{baudon2013decomposing}. Baudon {\it et al.} proved that the problem of determining
the irregular chromatic index of a graph can be handled in linear
time when restricted to trees \cite{bensmail2013complexity}.
It is then natural to ask if the same holds for bipartite graphs.
The following problem which was also asked by Baudon {\it et al.} in \cite{bensmail2013complexity} remains unsolved.

\begin{prob} [Baudon {\it et al.} \cite{bensmail2013complexity}]
Determine the computational complexity of deciding whether $\chi'_{irr}=2$ for bipartite graphs.
\end{prob}

Baudon {\it et al.}
characterized all connected graphs which cannot be decomposed into locally
$1$-irregular subgraphs and  call them exceptions \cite{baudon2013decomposing}.
For each $k> 1$, it would be interesting to characterize
all connected graphs which cannot be decomposed into locally
$k$-irregular subgraphs.

\begin{prob}
For each $k$, characterize
all connected graphs which cannot be decomposed into locally
$k$-irregular subgraphs.
\end{prob}

\begin{prob}
For each $k$, can one decide in polynomial time whether a given graph $G$, can be decomposed into locally
$k$-irregular subgraphs?
\end{prob}

Baudon {\it et al.} conjectured $h(1)\leq 3$ \cite{baudon2013decomposing}.
Here, we proved that $h(k) \geq 2k+1$ and  $h(k)=\Omega(k^2)$. The next question is to determine whether there is a function
$f$ in terms of $k$ such that $h(k) \leq f(k)$.

\begin{prob}\label{problem1}
Does there is a  function
$f$ in terms of $k$ such that $h(k) \leq f(k)$?
\end{prob}

In this work, by using mutually orthogonal Latin squares we proved that  $h(k)=\Omega(k^2)$. Finding a better lower bound can
be interesting.

It was shown that determining whether $reg(G)\leq \Delta(G)$ for a given connected graph $ G $  with minimum degree one
is {\bf NP}-complete \cite{reg}. Does there exist a graph $G$ with minimum degree two, such that $reg(G)= \Delta(G)+1$?

\begin{prob}
For each $t$, does there exist a graph $G$ with minimum degree $t$, such that $reg(G)= \Delta(G)+1$?
\end{prob}

We proved that every graph $G$ can be decomposed into $\Delta(G)$
subgraphs such that each subgraph locally regular and this bound is sharp for trees.
Characterizing all connected graphs which cannot be decomposed into $\Delta(G)-1$
subgraphs such that each subgraph is locally regular is interesting.

\begin{prob}
Characterize all connected graphs which cannot be decomposed  into $\Delta(G)-1$
subgraphs such that each subgraph is locally regular.
\end{prob}

Every planar graph $G$ with degree set $\{2, 4\}$ can be decomposed into two regular subgraphs.
We proved that
determining whether $ reg(G) = 2 $ for a given planar graph $G$ with degree set $\{3, 6\}$
is {\bf NP}-complete. With a similar technique we can show that deciding whether $ reg(G) = 2 $ for a
given planar graph $G$ with degree set $\{5, 10\}$
is {\bf NP}-complete. Since there is no $6$-regular planar graph, the only remaining case
is planar graphs with degree set $\{4, 8\}$.

\begin{prob}
Determine the computational complexity of deciding whether $ reg(G) = 2 $ for
planar graphs with degree set $\{3, 6\}$.
\end{prob}

In this work, for every $k>2$, we constructed a tree $T$ with $\Delta(T)= k$ such that
$T$ cannot be decomposed into a matching and
a locally irregular subgraph and also, we proved that every tree
can be decomposed into two matchings and a locally irregular subgraph. Computational
complexity of determining whether a given tree
can be decomposed into a matchings and a locally irregular subgraph remains open.

\begin{prob}
Determine the computational complexity of deciding whether a given tree
can be decomposed into a matching and a locally irregular subgraph.
\end{prob}

Regarding the error that we found in \cite{bensmail2013complexity}, we should mention that in order to preserve the planarity it would be sufficient to show
that if
all variables are inside a clause cycle (a cycle which connects all clauses in an arbitrary order)
 and the graph is planar, then 1-in-3
Satisfiability is $ \mathbf{NP}$-complete.
It is interesting to mention that although
Planar 3-SAT is $ \mathbf{NP}$-complete if all clauses are inside the variable cycle
and the graph is planar, then the problem is polynomial-time solvable \cite{dem}.
A special case is when all clauses are also connected in a path. Then the problem is
still in $ \mathbf{P}$ because we can show that this implies that clauses are all on one side of the variable cycle \cite{dem}.
We conjecture that the 1-in-3
Satisfiability in that case is also polynomial-time solvable.

\section{Acknowledgment}

The work of the second author was done during a sabbatical at the School of
Mathematics and Statistics, Carleton University, Ottawa. The second author is grateful to Brett Stevens for hosting this
sabbatical.

\small

\bibliographystyle{plain}
\bibliography{luckyref}

\end{document}